\documentclass[reqno,12pt]{amsart}

\usepackage{color} % black,white,red,green,blue,cyan,magenta,yellow
\usepackage{ifpdf}
\ifpdf
    \usepackage[pdftex]{graphicx}
    \usepackage[pdftex]{hyperref}
    \hypersetup{
        unicode=false,          % non-Latin characters in Acrobat’s bookmarks
        pdftoolbar=true,        % show Acrobat’s toolbar?
        pdfmenubar=true,        % show Acrobat’s menu?
        pdffitwindow=false,     % window fit to page when opened
        pdfstartview={FitH},    % fits the width of the page to the window
        pdftitle={MCP Article},      % title
        pdfauthor={Michael Holst},   % author
        pdfsubject={Mathematics},    % subject of the document
        pdfcreator={Michael Holst},  % creator of the document
        pdfproducer={Michael Holst}, % producer of the document
        pdfkeywords={PDE, analysis, mathematical physics}, % list of keywords
        pdfnewwindow=true,      % links in new window
        colorlinks=true,        % false: boxed links; true: colored links
        linkcolor=red,          % color of internal links
        citecolor=blue,         % color of links to bibliography
        filecolor=magenta,      % color of file links
        urlcolor=cyan           % color of external links
    }

\else
    \usepackage{graphicx}
    \usepackage{pstricks}
    
    \newcommand{\href}[2]{#2}
\fi

\usepackage{times}
\usepackage{amsfonts}

\usepackage{amsmath}
\usepackage{amsthm}
\usepackage{amssymb}
\usepackage{amsbsy}
\usepackage{amscd}

\usepackage{enumerate}
\usepackage{verbatim}
\usepackage{subfigure}

\newtheorem{theorem}{Theorem}[section]
\newtheorem{corollary}[theorem]{Corollary}

\numberwithin{equation}{section}  %amsmath command: tie counter to section

  \newcounter{mnote}
  \let\oldmarginpar\marginpar
    \renewcommand\marginpar[1]{\-\oldmarginpar[\raggedleft\footnotesize #1]%
    {\raggedright\footnotesize #1}}
\definecolor{myblue}{rgb}{0.2,0.2,0.7}
\definecolor{mygreen}{rgb}{0,0.6,0}
\definecolor{mycyan}{rgb}{0,0.6,0.6}
\definecolor{myred}{rgb}{0.9,0.2,0.2}
\definecolor{mymagenta}{rgb}{0.9,0.2,0.9}
\definecolor{mywhite}{rgb}{1.0,1.0,1.0}
\definecolor{myblack}{rgb}{0.0,0.0,0.0}

\newcommand{\beq}{\begin{equation}}
\newcommand{\eeq}{\end{equation}}
\newcommand{\beqa}{\begin{eqnarray}}
\newcommand{\eeqa}{\end{eqnarray}}
\newcommand{\IF}{I\!\!F}         % Fermi transport symbol.

\newcommand{\geqs}{\geqslant}      % Nice =<
      
\newcommand{\FETK}{{\sc FETK}}

\usepackage{jen_macros}  %my version in separate file
\begin{document}

\title[Barrier Methods for Critical Exponent Problems]
      {Barrier Methods for
       Critical Exponent Problems in 
       Geometric Analysis and Mathematical Physics}

\author[J. Erway]{Jennifer Erway}
\email{erwayjb@wfu.edu}
\address{Department of Mathematics\\
         Wake Forest University \\
         Winston-Salem, NC 27109}

\author[M. Holst]{Michael Holst}
\email{mholst@math.ucsd.edu}

\address{Department of Mathematics\\
         University of California San Diego\\ 
         La Jolla CA 92093}

\thanks{MH was supported in part by NSF Awards~0715146 and 0915220,
and by DOD/DTRA Award HDTRA-09-1-0036.}
\thanks{JE was supported in part by a Ralph E. Powe Junior Faculty Enhancement Award.}

\date{\today}


\begin{abstract}
We consider the design and analysis of numerical methods
for approximating positive solutions to nonlinear geometric elliptic 
partial differential equations containing critical exponents. 
This class of problems includes the Yamabe problem and the Einstein 
constraint equations, which simultaneously contain several challenging 
features: high spatial dimension $n \geqs 3$, varying (potentially 
non-smooth) coefficients, critical (even super-critical) nonlinearity, 
non-monotone nonlinearity (arising from a non-convex energy), and spatial
domains that are typically Riemannian manifolds rather than simply
open sets in $R^n$.
These problems may exhibit multiple solutions, 
although only positive solutions typically have meaning.
This creates additional complexities in both the theory and numerical 
treatment of such problems, as this feature introduces both non-uniqueness
as well as the need to incorporate an inequality constraint into the
formulation.
In this work, we consider numerical methods based on Galerkin-type 
discretization, covering any standard bases construction (finite element, 
spectral, or wavelet), and the combination of a barrier method for
nonconvex optimization and global inexact Newton-type methods for dealing
with nonconvexity and the presence of inequality constraints.
We first give an overview of barrier methods in non-convex optimization,
and then develop and analyze both a primal barrier energy method 
for this class of problems.
We then consider a sequence of numerical experiments using this type
of barrier method, based on a particular Galerkin method, namely the 
piecewise linear finite element method, leverage the FETK modeling package.
We illustrate the behavior of the primal barrier energy method for
several examples, including the Yamabe problem and the Hamiltonian constraint.
\end{abstract}

\newcommand{\Astrut}{\rule[-.5ex]{0pt}{2ex}}%
\newcommand{\struta}{\rule[-1.25ex]{0pt}{2ex}}%
\newcommand{\strutb}{\rule[-.9ex]{0pt}{2ex}}%
\newcommand{\mystrut}{\vrule height9.5pt depth1.5pt width0pt}

\newcommand{\agap}{\hspace{1.5ex}}
\newcommand{\AlgSec}[1]{{\bf[#1.]}\vrule height10pt depth1.5pt width0pt}

\newcommand{\tab}{\par\noindent\mystrut}
\newcommand{\tabb}{\tab\hskip1.5em}
\newcommand{\tabbb}{\tab\hskip3.5em}
\newcommand{\tabbbb}{\tab\hskip5.5em}
\newcommand{\tabbbbb}{\tab\hskip7.5em}
\newcommand{\tabbbbbb}{\tab\hskip9.5em}

\newcommand{\alg}[1]{\par\noindent\mystrut\ignorespaces\hbox to\textwidth{#1\hfill}}
\newcommand{\algt}[1]{\par\noindent\mystrut\hbox to\textwidth{\ignorespaces\hskip1.5em#1\hfill}}
\newcommand{\algtt}[1]{\par\noindent\mystrut\hbox to\textwidth{\ignorespaces\hskip3.5em#1\hfill}}
\newcommand{\algttt}[1]{\par\noindent\mystrut\hbox to\textwidth{\ignorespaces\hskip5.5em#1\hfill}}
\newcommand{\algtttt}[1]{\par\noindent\mystrut\hbox to\textwidth{\ignorespaces\hskip7.5em#1\hfill}}
\newcommand{\algttttt}[1]{\par\noindent\mystrut\hbox to\textwidth{\ignorespaces\hskip9.5em#1\hfill}}
\newcommand{\clg}[2]{\par\noindent\mystrut\hbox to\textwidth{\ignorespaces#1\hfill[#2]}}
\newcommand{\clgt}[2]{\par\noindent\mystrut\hbox to\textwidth{\ignorespaces\hskip1.5em#1\hfill[#2]}}
\newcommand{\clgtt}[2]{\par\noindent\mystrut\hbox to\textwidth{\ignorespaces\hskip3.5em#1\hfill[#2]}}
\newcommand{\clgttt}[2]{\par\noindent\mystrut\hbox to\textwidth{\ignorespaces\hskip5.5em#1\hfill[#2]}}
\newcommand{\clgtttt}[2]{\par\noindent\mystrut\hbox to\textwidth{\ignorespaces\hskip7.5em#1\hfill[#2]}}
\newcommand{\clgttttt}[2]{\par\noindent\mystrut\hbox to\textwidth{\ignorespaces\hskip9.5em#1\hfill[#2]}}

\newcommand{\tac}{\par\noindent\mystrut\hskip1em}
\newcommand{\tabc}{\tab\hskip2.5em}
\newcommand{\tabcc}{\tab\hskip4.5em}
\newcommand{\tabccc}{\tab\hskip6.5em}
\newcommand{\tabcccc}{\tab\hskip8.5em}
\newcommand{\tabccccc}{\tab\hskip10.5em}

\newcommand{\REPEAT}{\textbf{repeat}\Astrut}
\newcommand{\WHILE}{\textbf{while}\hskip2pt}
\newcommand{\FOR}{\textbf{for}\hskip2pt}
\renewcommand{\IF}{\textbf{if}\hskip2pt}

\newcommand{\logical}[1]{\{\mbox{#1}\}}
\newcommand{\proc}[1]{\setbox4=\hbox{\noindent\strut#1}\boxit{\box4}}
\newcommand{\procc}[1]{\setbox4=\vbox{\hsize 20pc \noindent\strut#1}\boxit{\box4}}
\newcommand{\DO}{\hskip2pt\textbf{do}\hskip2pt\Astrut}
\newcommand{\UNTIL}{\textbf{until}\hskip2pt}
\newcommand{\BEGIN}{\textbf{begin}}
\newcommand{\END}{\textbf{end}}
\newcommand{\ENDIF}{\textbf{end if}}
\newcommand{\ENDWHILE}{\textbf{end while}}
\newcommand{\ENDREPEAT}{\textbf{end repeat}}
\newcommand{\ENDDO}{\textbf{end do}}
\newcommand{\ENDFOR}{\textbf{end (for)}}
\newcommand{\THEN}{\hskip2pt\textbf{then}\hskip2pt}
\newcommand{\ELSE}{\hskip2pt\textbf{else}\hskip2pt}
\newcommand{\FIRSTELSE}{\textbf{else}\hskip2pt}
\newcommand{\ELSEIF}{\textbf{else if}}
\newcommand{\RETURN}{\textbf{return}}
\newcommand{\FIRST}{\hskip-2pt}  % \FIRST\ELSE if \ELSE starts a line

\newcommand{\STOP}{\mathbf{stop}}
\newcommand{\BREAK}{\mathbf{break}}
\newcommand{\TRUE}{\mathbf{true}}
\newcommand{\FALSE}{\mathbf{false}}
\newcommand{\true}{\mathbf{true}}
\newcommand{\false}{\mathbf{false}}
\newcommand{\NOT}{\mathop{\mathbf{not}\,}}
\newcommand{\AND}{\mathop{\;\mathbf{and}\;}}
\newcommand{\OR}{\mathop{\;\mathbf{or}\;}}

% Pseudo-code only macros

\newcommand{\badSteps}{\mathit{badSteps}}
\newcommand{\converged}{\mathit{converged}}
\newcommand{\complete}{\hbox{\it complete}}
\newcommand{\CSpoint}{\hbox{\it subspace\_stationary\_pt}}
\newcommand{\exchange}{\mbox{\it exchange}}
\newcommand{\exit}{\hbox{\it exit}}
\newcommand{\found}{\hbox{\it found}}
\newcommand{\feasible}{\mathit{feasible}}
\newcommand{\hitcon}{\mbox{\it hit\_constraint}}
\newcommand{\addcon}{\mbox{\it add\_constraint}}
\newcommand{\Hsingular}{\mbox{\it singular\_H}}
\newcommand{\imprvd}{\hbox{\it improved}}
\newcommand{\indef}{\hbox{\it indefinite}}
\newcommand{\minimizer}{\mbox{\it subspace_minimizer}}
\newcommand{\optimal}{\mathit{optimal}}
\newcommand{\posdef}{\hbox{\it positive\_definite}}
\newcommand{\unbndd}{\hbox{\it unbounded}}
\newcommand{\semidef}{\hbox{\it positive\_semidefinite}}
\newcommand{\stationary}{\hbox{\it stationary\_point}}
\newcommand{\subspace}{\hbox{\it subspace\_convergence}}
\newcommand{\singular}{\mathit{singular}}
\newcommand{\tol}{\mathit{tol}}
\newcommand{\uncon}{\mbox{\it unconstrained\_step}}
\newcommand{\vertex}{\mbox{\it vertex}}
\newcommand{\wschosen}{\hbox{\it working\_set\_chosen}}
\newcommand{\Wsingular}{\mbox{\it singular\_W}}
\newcommand{\Asingular}{\mbox{\it singular\_A}}

%%% Local Variables:
%%% mode: latex
%%% TeX-master: t
%%% End:

\maketitle

%\clearpage

\vspace*{-0.95cm}
{\footnotesize
\tableofcontents
}
\vspace*{-0.5cm}

\newcommand{\tWhite}[1]{{\color{white}#1}}
\newcommand{\defined}{\mathop{\,{\scriptstyle\stackrel{\triangle}{=}}}\,}
\newcommand{\words}[1]{\mgap\text{#1}\mgap}
\newcommand{\mgap}{\;\;}
\newcommand{\subject}{\hbox{\rm subject to}}
\newcommand{\minimize}[1]{{\displaystyle\minim_{#1}}}
\newcommand{\minim}{\mathop{\mathrm{minimize}}}
\newcommand{\diag}{\mathop{\mathrm{diag}}}
\newcommand{\strict}{\mathop{\mathrm{strict}}}
\newcommand{\interior}{\mathop{\mathrm{int}}}
\newcommand{\st}{\mathop{\mathrm{ : }}}
\newcommand{\BigO}[1]{O\big(#1\big)}
\newcommand{\strt}{\rule[-.5ex]{0pt}{3ex}}%   Total height 3.5                                                                      
\newcommand{\hstrt}{\rule[-1ex]{0pt}{3.5ex}}% Total height 4.5                                                                      
\newcommand{\CG}{{\small CG}}

%%%%%%%%%%%%%%%%%%%%%%%%%%%%%%%%%%%%%%%%%%%%%%%%%%%%%%%%%%%%%%%%%%%%%%%%%%%%%%
\section{Introduction}
\label{sec:intro}

In this article we consider the design and analysis of numerical methods
for approximating positive solutions to nonlinear geometric elliptic 
partial differential equations containing critical exponents. 
These types of problems arise regularly in geometric analysis and 
mathematical physics, examples of which include the Yamabe problem
and the Einstein constraint equations~\cite{HNT07a,HNT07b}.
These problems often simultaneously contain several challenging features, 
including spatial dimension $n \geqs 3$, varying and potentially non-smooth
coefficients, critical (or even super-critical) nonlinearity, 
non-monotone nonlinearity (arising from a nonconvex energy), and spatial
domains that are typically Riemannian manifolds rather than simply
open sets in $R^n$.
For these types of problems, there may be multiple solutions, although
only positive solutions typically have mathematical and physical meaning.
This creates additional complexities in both the theory and numerical 
treatment of such problems, as this feature introduces both non-uniqueness
as well as the need to incorporate an inequality constraint into the
formulation.
In this work, we consider numerical methods based on Galerkin-type 
discretization, covering any standard bases construction (finite element, 
spectral, or wavelet), and the combination of a barrier method for
nonconvex optimization and global inexact Newton-type methods for dealing
with nonconvexity and the presence of inequality constraints.
Our goal is to develop reliable methods for computing positive approximate
solutions to these types of nonlinear problems.

Critical exponent problems arise in a fundamental way throughout
geometric analysis and mathematical general relativity.
One of the seminal problems in this area is the {\em Yamabe Problem}:
Find $u \in X$ such that
\begin{align}\label{eqn-yamabe}
-8\Delta_g u + R u & = R_u u^5 \quad \mbox{ in } \Omega,
\\
u & > 0, 
\end{align}
%\end{eqnarray}
where $\Omega$ is a Riemannian $3$-manifold,
$g$ is the positive definite metric on $\Omega$,
$\Delta_g$ is the Laplace-Beltrami operator generated by $g$,
$R$ is the scalar curvature of $g$,
and $R_u$ is the scalar curvature corresponding to the
{\em conformally transformed} metric:
\begin{equation}
\overline{g} = \phi^4 g.
\end{equation}
The coefficients $R$ and $R_u$ can take any sign.  The Banach space
$X$ containing the solution is an appropriate Sobolev class
$W^{s,p}(\Omega)$ for suitably chosen exponents $s$ and $p$.  If the
manifold $\Omega$ has a boundary, then boundary conditions are also
prescribed, such as $u=1$ on an exterior boundary to $\Omega$.  In the
case that $\Omega \subset \mathbb{R}^3$, and $g_{ij}=\delta_{ij}$,
then $\Delta_g$ reduces to just the Laplace operator on $\Omega$.
With the presence of the term $u^5$ and the spatial dimension
being three, this is an example of a {\em critical exponent problem};
such problems are known to be difficult to analyze as well as to
simulate numerically.  The presence of the inequality constraint (only
positive solutions have mathematical and physical meaning) creates
additional complexities in both the theory and numerical treatment of
such problems.  Prior work on numerical methods for critical exponent
semilinear problems has focused primarily on the development of
adaptive methods for recovering solution blowup;
cf.~\cite{Budd.C;Humphries.A1998,Budd.C1998}.

{\em Outline of the paper.}
The structure of the remainder of the paper is as follows.
In \S\ref{sec:einstein}, we give a more detailed overview of the
class of geometric PDE problems of interest, including both the
Yamabe problem and the Hamiltonian constraint in the Einstein equations.
As part of the discussion, we derive the linearized Hamiltonian constraint,
and construct an artificial nonconvex ``energy'' functional, which gives 
rise to the Hamiltonian constraint as a condition for its stationarity.
In \S\ref{sec:methods},
we give an overview of barrier methods in nonconvex optimization.
In \S\ref{sec:primal} we then develop and analyze a primal barrier 
energy method for this class of 
problems.
Finally, in \S\ref{sec:num}
we consider a sequence of numerical experiments using this type
of barrier method, based on a particular Galerkin method, namely the 
piecewise linear finite element method, leverage the FETK modeling package.
We illustrate the behavior of the primal barrier energy method for
several examples, including the Yamabe problem and the Hamiltonian constraint.
We draw some conclusions in \S\ref{sec:conc}.

%%%%%%%%%%%%%%%%%%%%%%%%%%%%%%%%%%%%%%%%%%%%%%%%%%%%%%%%%%%%%%
\section{Elliptic Problems in Geometric Analysis and Relativity}
\label{sec:einstein}

While one of our motivations here is to develop methods for the Yamabe
problem and similar problems arising in geometric analysis, we are also
interested in a related, more general problem arising in general relativity.
The Einstein equations, which represents Einstein's 1915 theory of gravity,
are a coupled hyperbolic-elliptic system that governs the deformation of
the underlying metric of spacetime in response to the distribution and
dynamics of matter and energy density.
The elliptic part of the system, known as the 
{\em Einstein constraint equations},
or the coupled {\em Hamiltonian and momentum constraints},
are of great interest in both the mathematical and numerical
relativity research communities.
This elliptic system must be satisfied by initial data used to
evolve the metric forward in time with the hyperbolic portion of
the Einstein equations (called the evolution equations), and the
constraints must also be satisfied at all points in time during
the evolution.

The Einstein constraints have all of the difficult features of the
Yamabe problem, plus more: three spatial dimensions, non-flat
Riemannian manifold spatial domain, critical exponent, non-monotone
nonlinearity, negative exponent powers (non-polynomial rational nonlinearity,
giving rise to singularities at the origin), possibly non-smooth 
coefficients, possible non-uniqueness,
physical positivity requirement, and the structure
of an elliptic system for two variables, with lack of a variational
structure (the equations do not arise as the Euler condition for
stationarity of an underlying energy functional).
However, one of the difficulties is not present in an important physical
situation known as the {\em constant mean curvature (CMC) case}:
in this situation,
the coupled elliptic system of the Hamiltonian and momentum constraints
{\em decouple} into the separate constraints, both of which have
separate variational structure.
However, all of the other difficulties remain, and the Hamiltonian
constraint alone may be viewed as a generalization of the Yamabe problem.
An overview of the Einstein constraints, including the CMC case, can
be found in~\cite{HNT07a,HNT07b}.
Here we will consider only the CMC case, and focus on the
Hamiltonian constraint, also known as the {\em Lichnerovich equation}.
\begin{eqnarray} 
-\nabla\cdot\left(\bar{a}\nabla u\right) + \frac{R}{8}u & = & -\frac{\tau ^2}{12}u^5+
\frac{\sigma ^2}{8} u^{-7} + 2\pi\rho u^{-3} \quad \mbox{in $\Omega$},\label{eqn-strong1} \\
\left(\bar{a}\nabla u(x,y,z)\right) \cdot n + cu& = & g_N
\hspace{1cm}\mbox{on $\partial_N\Omega$}, \label{eqn-strong2}\\
u & = & g_D \hspace{1cm} \mbox{on $\partial_D \Omega$},  \label{eqn-strong3}
\end{eqnarray}
where $n$ is the unit normal and $\partial \Omega=\partial_D \Omega \bigcup \partial_N \Omega$ and 
$\partial_D \Omega \bigcap \partial_N \Omega = \emptyset$.  Here $R(x)$, $\tau^2
(x)$, $\sigma ^2(x):\Omega\subset \Re^3\rightarrow\Re$ and $\rho$,
$\sigma^2$, $\tau^2
\geq 0$ for all $x\in\Omega$.  Also it assumed that there exist positive 
constants $C_1$ and $C_2$ such that $-C_1\leq R(x) \leq C_2$.  
Also, $\tau$ may be considered to be constant so that CMC decoupling occurs.
Reasonable values for $\sigma$ are such that $0\geq \sigma ^2\leq C_3$
where $C_3$ may be as large as $10^6$.  
The strong form of the constraints given in
(\ref{eqn-strong1}--\ref{eqn-strong3}) can be transformed into the weak
form, reformulating the problem using fewer derivatives.  

\subsection{Weak Formulation} \label{chap:weak} The weak formulation is
obtained by taking the $L^2$-inner product over $\Omega$ with all test
functions $v\in H_0^1(\Omega)\defined\left\{u\in H^1(\Omega) \, : \,  v=0 \words{on} \Omega_D \right\}$ and (\ref{eqn-strong1}), yielding:
\begin{eqnarray*}
\int_{\Omega}\left(-\bar{a}\triangle u + \frac{R}{8}u + \frac{\tau ^2}{12}u^5-
\frac{\sigma ^2}{8} u^{-7} - 2\pi\rho u^{-3}\right)v\,\,\, dx = 0. 
\end{eqnarray*}
Green's first identity states:
\begin{eqnarray*}
\int_{\Omega}(\nabla\cdot z)v \,\,\, dx = \int_{\partial\Omega}(n\cdot z)
v\,\,\, ds - \int_{\Omega}z\cdot\nabla v \,\,\, dx.
\end{eqnarray*}
Taking $z :=\bar{a}\nabla u $ and recalling that $v=0$ on $\partial_D(\Omega)$,
we obtain:
\begin{equation}\label{eqn-weak-progress}
 \int_{\Omega}\left(\bar{a}\nabla u\right) \cdot\nabla v\,\,\, dx - \int_{\partial_N\Omega}(n\cdot \bar{a}
\nabla u )v\,\,\, ds  + \int_{\Omega} k(u)v \,\,\,dx = 0,
\end{equation}
where $k(u)\defined \frac{R}{8}u + \frac{\tau ^2}{12}u^5-\frac{\sigma
  ^2}{8} u^{-7} - 2\pi\rho u^{-3}$.  Using (\ref{eqn-strong2}) in (\ref{eqn-weak-progress}), yields:
\begin{eqnarray*}
 \int_{\Omega}\left(\bar{a}\nabla u\right)\cdot\nabla v + k(u)v\,\,\, dx -
 \int_{\partial_N\Omega}\left( g_N -cu\right) v\,\,\, ds = 0.
\end{eqnarray*}
Thus, the weak form of 
(\ref{eqn-strong1}-\ref{eqn-strong3})
is given by:
\begin{equation}\label{eqn-weak1}
\text{Find}~ u\in X = H_{D}^1(\Omega ) \cap [u_-,u_+] ~\text{s.t.}~
(f(u),v)=0, ~\text{for all}~ v\in H_0^1(\Omega),
\end{equation}
where
 $H_{D}^1(\Omega )\defined\left\{ u\in H^1(\Omega) \, : \,
  u=g_D \words{on} \Omega_D\right\}$ and 
\begin{equation}\label{eqn-weak2}
(f(u),v)\defined\int_{\Omega}\left[ \left(\bar{a}\nabla u\right)\cdot\nabla v
+k(u)v \right] \,\,\, dx
 +\int_{\partial_N\Omega}(c u - g_N) v.
\end{equation}
Note that we have constructed the space 
$X = H_{D}^1(\Omega ) \cap [u_-,u_+]$ 
in which to look for solutions to the problem based on the need to ``guard''
the nonlinearity from blow-up at the origin.
This potential blowup is due to the negative powers appearing as part of 
the non-polynomial, rational form of the nonlinearity.
The pointwise interval $[u_-,u_+]$, which can be strictly negative
or strictly positive, can be shown to contain one or more solutions
using maximum principles and fixed-point arguments; cf.~\cite{HNT07a,HNT07b}.
The numerical methods we develop later in the paper will incorporate
this type of ``guarding'' in the discrete formulation.

%%%%%%%%%%%%%%%%%%%%%%%%%%%%%%%%%%%%%%%%%%%%%%%%%%%%%%%%%%%%%%%

The Gateaux derivative $(Df(u)w,v)$ of the weak nonlinear form
(\ref{eqn-weak2}) is needed for use in Newton-like algorithms.  It is
computed formally as
\begin{eqnarray}
\frac{\partial}{\partial t}f\left(u+t w\right)(v)
\bigg\vert_{t=0} &= &
\frac{\partial}{\partial t}
\left[\int_{\Omega}
\left[ \left( \bar{a}\nabla (u+t w)\right)\cdot \nabla v + k(u+t w)
    v\right]  \,\,\, dx \right] \bigg\vert_{t=0} \nonumber
\\ & & 
+\frac{\partial}{\partial t}\left[\int_{\partial_N\Omega} 
 \left( c \left(u+t w\right) - g_N \right)v    \,\,\, ds
\right]
 \bigg\vert_{t=0}, \nonumber \\
\end{eqnarray}
giving that
\begin{eqnarray}
(Df(u)w,v) =
\int_{\Omega}
\left[ \left( \bar{a}\nabla w\right)\cdot \nabla v + k'(u)wv\right]  \,\,\, dx
+
\int_{\partial_N\Omega} 
  cwv \,\,\, ds,  \label{eqn-second-deriv}
\end{eqnarray}
where $k'(u)wv=\frac{R}{8}wv+\frac{5}{12}\tau ^2u^4wv+
\frac{7}{8}\sigma ^2 u^{-8}wv + 6\pi\rho u^{-4}wv$.

%%%%%%%%%%%%%%%%%%%%%%%%%%%%%%%%%%%%%%%%%%%%%%%%%%%%%%%%%%%%%%%%%%%%%%%%%%
\subsection{The energy} \label{chap:energy} The weak formulation can be
viewed as a zero-finding problem; alternatively, it may be viewed as the
problem of finding a critical point of an energy functional.

\begin{theorem}\label{thrm-cp}
  $u$ is a solution to the weak form (\ref{eqn-weak1}) if and only if $u$
  is a critical point of the energy functional
\begin{eqnarray}\label{eqn-energy}
  J(u)=\int_{\Omega}\left[\frac{1}{2}\left( \bar{a}\nabla u\right)\cdot\nabla u 
    + \frac{1}{16}Ru^2+\frac{1}{72} \tau^2 u^6
    +\frac{1}{48}\sigma^2u^{-6}+\pi\rho u^{-2}\right] \,\,\, dx \nonumber \\
  +\frac{1}{2}\int_{\partial_N\Omega} cu^2 \,\,\, ds - \int_{\partial_N\Omega} g_N
  u \,\,\, ds.
\end{eqnarray}
\end{theorem}
\begin{proof}
We prove this theorem by showing that Gateaux derivative of $J(u)$ is exactly the weak formulation (\ref{eqn-weak1}).
We have 
\begin{eqnarray}
  \frac{\partial}{\partial t}J\left(u+t v\right)
  \bigg\vert_{t=0} &=& 
  \frac{\partial}{\partial t}\left[
    \int_{\Omega}
    \left[\frac{1}{2}\left( \bar{a}\nabla (u+tv)\right)\cdot\nabla (u+tv) 
      + \frac{1}{16}R(u+tv)^2 \right.\right.\nonumber \\
  & & \left. \left. + \frac{1}{72} \tau^2 (u+tv)^6 
      +\frac{1}{48}\sigma^2(u+tv)^{-6}+\pi\rho (u+tv)^{-2} \right] \,\,\, dx \right] \bigg\vert_{t=0} \nonumber \\ 
  & & 
  +\frac{\partial}{\partial t}\left[\int_{\partial_N\Omega} 
    \left[\frac{1}{2} c \left(u+t v\right)^2 - g_N (u+tv)\right]   \,\,\, ds \right]
  \bigg\vert_{t=0} \nonumber \\
  &=& 
  \int_{\Omega}\left[ \left(\bar{a}\nabla u\right)\cdot\nabla v
    +\frac{R}{8}uv + \frac{\tau ^2}{12}u^5v-\frac{\sigma
      ^2}{8} u^{-7}v - 2\pi\rho u^{-3}v
  \right] \,\,\, dx \nonumber \\
& &
  +\int_{\partial_N\Omega}(c u - g_N) v \,\,\, ds \nonumber \\
  &=&
  (f(u),v).\nonumber
\end{eqnarray}
\end{proof}
Theorem~\ref{thrm-cp} can be restated in terms of stationarity.
\begin{corollary}
  The energy functional $J(u)$ in (\ref{eqn-energy}) is stationary at $u$
  if and only if $u$ is a solution to the weak formulation
  (\ref{eqn-weak1}).
\end{corollary}

From an optimization standpoint, it is reasonable to ask under what
conditions can $J(u)$ be minimized to obtain the solution to the weak
formulation.  Certainly, if $J(u)$ is a convex energy functional then a
weak solution can be found by minimizing $J(u)$.  Convexity implies
uniqueness of solutions, which in this application is often not the case.
Nevertheless, Newton's method may be used to find a critical point of the
energy, requiring the need for a second derivative.  From
Theorem~\ref{thrm-cp}, the second Gateaux derivative of the energy
functional $J(u)$ is exactly $D(f(u),v)$ (see (\ref{eqn-second-deriv})).

\textbf{Example 1}.  Set $\bar{a}=1.0$, $R=1.0$, $\tau=0.1$, $\sigma=0.2$,
$\rho=0.1$, $c=1.0$, and $g_N=-1.0$ with $\Omega$ chosen to be a single
hole domain at the origin in three-dimensional space with only Robin
boundary conditions.  In this simple case, since $R\ge 0$ and $c\ge 0$, the
energy functional $J(u)$ is convex on the domain of positive functions
$u>0$, and thus, minimizing the energy functional $J$ over $u>0$ is
equivalent to solving the weak formulation.  Similarly, if $R\ge 0$ and $c\ge
0$, then $J(u)$ is also convex over $u<0$--allowing for the existence of
both strictly negative and positive solutions.  This yields a convex energy
functional $J$ with one positive solution $u>0$ and one negative solution
$u<0$ (see Section~\ref{sec:num}).

\textbf{Example 2}. Set $a=2.0$, $\tau=\sqrt{72}$, $\sigma=\sqrt{48}$, and
$\rho=1/\pi$, and let $\Omega$ be the one-dimensional closed subset
$[0.1,10]$.  With these choices, the energy reduces to:
\begin{equation} \label{eqn-energy-1d}
J(u)=\int_{0.1}^{10}\left[\left( \nabla u\right)\cdot\nabla u +\frac{1}{16}Ru^2+u^6+u^{-6}+ u^{-2}\right] \,\,\, dx .
\end{equation}
Let $I(u)$ denote the integrand in (\ref{eqn-energy-1d}).  For $u\defined
x$, the second derivative of $I(u)$ is given by
$$
\frac{d^2 I(u)}{dx^2}=\frac{1}{8}R+30u^4+ 42u^{-8}+6u^{-4}.
$$
Thus, if $R<0$ and sufficiently negative, $I$ has an inflection point;
otherwise, $I$ is a strictly positive, convex function of $u$.  For
any value of $R$, $I(u)$ is an even function since $I(u)=I(-u)$.  The
integrand $I(u)$ with $R=-1000$ is plotted in Figure \ref{fig-1} in
coordinate pairs $(x,I(u))$ with $u\defined x$  on the interval $[0.4,3]$.  The second
derivative of $I(u)$ confirms this is a nonconvex function of $u$.
\begin{figure}[t]
\begin{center}
\includegraphics[height=2in]{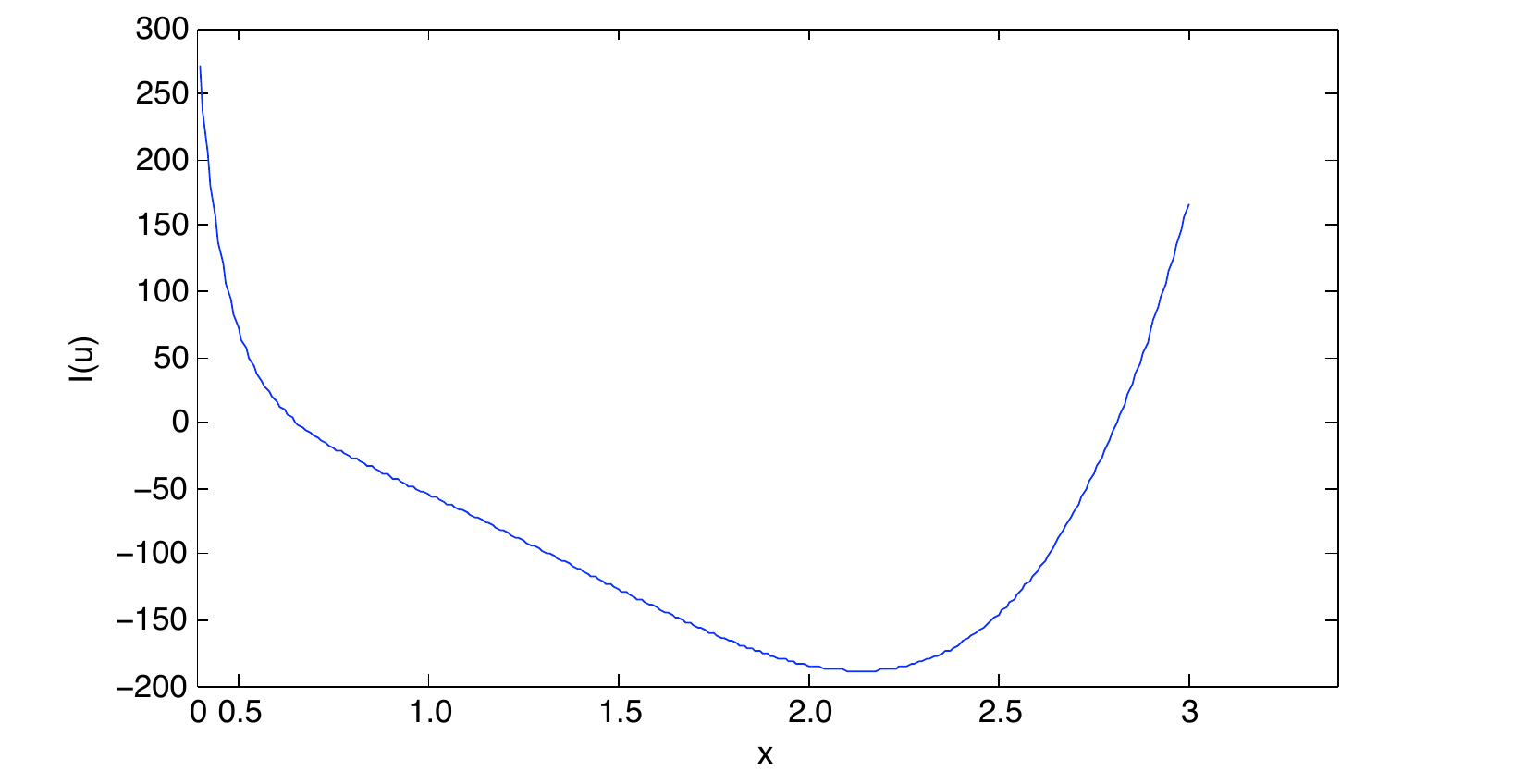}
\end{center}
\caption{A plot of the integrand $I(u)$ with $u\defined x$, $\bar{a}=2$, $R=-1000$,
  $\tau=\sqrt{72}$, $\sigma=\sqrt{48}$, $\rho=1/\pi$, and $\Omega=[0.4,3]$.}
\label{fig-1}
\end{figure}

When the energy functional is convex, a weak solution can be found by
minimizing the energy functional using a Newton-like iteration.  However,
when the energy is nonconvex, a weak solution may be a maximizer or
saddlepoint of the energy.  In this case, Newton's method can be used to
find a stationary point of the energy functional with progress towards a
stationary point being guaranteed by computing steps along the Newton
direction that yield sufficient decrease in a merit function.  The
positivity constraint $u>0$ can naturally be enforced when solving the
Lichnerovich equation (\ref{eqn-strong1}--\ref{eqn-strong3}) using a
safe-guarded Newton method.  Unlike the Lichnerovich equation, the Yamabe
Problem (\ref{eqn-yamabe}) does not have a singularity at $u=0$; thus, a
barrier method can be used to help enforce the positive inequality
constraint on $u$.  In the following section, we consider barrier methods
for nonconvex optimization and develop a primal barrier energy method
for this class of problems.

%%%%%%%%%%%%%%%%
\section{Barrier methods for nonconvex optimization}
\label{sec:methods}

Barrier methods are the most widely-used type of interior method for
general nonlinear inequality-constrained optimization problems of the form:
\begin{equation}                        \label{eqn-inequality}
 \begin{array}{lc}
   \minimize{x\in\Re^n}& f(x)\\
   \subject            & c(x)\ge 0,
 \end{array}
\end{equation}
where $f:\Re^n\rightarrow\Re$ is assumed to be twice-continuously
differentiable and $c:\Re^n\rightarrow\Re^m$ is
an $m$-vector of constraints. Generally speaking, barrier methods seek to
minimize a composite function that both resembles the original function and
naturally prevents infeasible iterates.  Today, the most widely used 
barrier function is the classical logarithmic barrier function:
\begin{equation}\label{eqn-frisch}
B_\mu(x)=f(x)-\mu \sum_{i=1}^m \ln c_i(x).
\end{equation}
Notice that if $\mu$ is small, the barrier function resembles the original
function; moreover, this function inherits the smoothness associated with
the original problem but is only defined in the strict interior of the
feasible region for the original problem (\ref{eqn-inequality}).

The classical barrier method solves (\ref{eqn-inequality}) by minimizing 
$B_\mu(x)$ for a decreasing sequence of positive $\mu$.  Given a fixed $\mu>0$,
first-order optimality conditions for $x^*$ to be a minimizer of the
barrier function is that $\nabla B_\mu(x^*)=0$, i.e.,
\begin{equation}\label{eqn-nablaB}
\nabla f(x^*)- J(x^*)^T y(x^*)=0,
\end{equation}
where $J$ denotes the constraint Jacobian and $y_i(x)\defined \mu / c_i(x)$.
Alternatively, $y$ is interpreted as a vector of Lagrange multipliers
associated with the original inequality problem (\ref{eqn-inequality}).
Moreover, $c_i(x) y_i =\mu$ can be viewed a perturbation of the
complementarity condition for a first-order KKT point.  

When $c(x)\defined x$, the Newton equations for minimizing $B_\mu (x)$ are
given by
\begin{equation} \label{eqn-barrierNewton}
\left(\nabla^2 f(x)+\mu \diag(X^{-2})\right) p = -\left[ \nabla f(x)-\mu X^{-1} \right],
\end{equation}
where $X^{-j}\defined [1/x_1^j, 1/x_2^j, \ldots, 1/x_n^j]^T$ for $j\in\{1,2\}$,
and $\diag(x)$ is the diagonal matrix whose $i$th diagonal entry is $x_i$.  To
ensure global convergence, a line search must be used to satisfy a
sufficient decrease criteria (e.g., the Armijo or Wolfe conditions).  Once
a minimizer $x(\mu)$ of $B_\mu(x)$ is computed, $\mu$ is reduced and the
process is repeated.  Notice that the subsequent minimization can be
warm-started by choosing the minimizer of the previous barrier function as the initial guess.
Algorithm~\ref{alg-barrier} summarizes the classical barrier method for
solving (\ref{eqn-inequality}) with $c(x)\defined x$ and using the Amijo condition
for sufficient descent.
\begin{figure}[htb]
\begin{center}
\begin{minipage}{\textwidth}
%%%%%%%%%%%%%%%%%%%%%%%%%%%%%%%%%%%%%%%%%%%%%%%%%%%%%%%%%%%%%%
\begin{Pseudocode}{Algorithm~\ref{alg-barrier}: Classical Barrier Method}\label{alg-barrier}
%\textsl{(This is the classical barrier method with simple bound constraints
%  on $x$ and the Armijo condition for sufficient decrease.)}
\tab Choose $x_0>0$, $\mu>0$, $\eta\in \left(0,\frac{1}{2}\right)$, and $\gamma\in (0,1)$;
\tab Set $k=0$;
\tab \WHILE{} $\NOT \converged$ \DO{}
\tab Compute $x(\mu)$, an unconstrained minimizer of $B_\mu(x)$:
\tabb \WHILE{} $\NOT \converged$ \DO{}
\tabbb Solve (\ref{eqn-barrierNewton}) to obtain $p$;
\tabbb Use the ``99\% rule'' to compute $\alpha_{\max}$;
\tabbb Compute $\alpha\in (0,\alpha_{\max}]$ such that $B_\mu(x+\alpha p) \leq B_\mu(x) +
\eta\alpha\nabla B_\mu(x)^Tp$;
\tabbb $x\leftarrow x+\alpha p$;
\tabb \ENDDO{}
\tabb $x_{k+1}\leftarrow x(\mu)$;
\tabb $\mu\leftarrow \gamma \mu$;
\tabb $k\leftarrow k+1$;
\tab \ENDDO{} 
\end{Pseudocode}
%%%%%%%%%%%%%%%%%%%%%%%%%%%%%%%%%%%%%%%%%%%%%%%%%%%%%%%%%%%%%%%%%%%%%                   
\end{minipage}
\end{center}
\end{figure}

When computing a step length, it is necessary to safeguard the step to
avoid taking a step into the infeasible region.  For linear constraints, the
so-called ``99\%'' rule may be invoked (e.g., see ~\cite{ForGW02}) that states that if
$\Delta x_i <0$ for at least one $i$: (i) Compute $\alpha_{max} = \max\{  
-(x_i/ \Delta x_i)  \,\, : \,\, \Delta x_i < 0 \}$; (ii) Set the
  maximum step length to be $\bar{\alpha} = \max\{ 0.99\times \alpha_{max}, 1
  \}$.  

We now state two theorems that summarize the convergence of the
classical barrier method.  The first theorem governs the local
convergence of a sequence of minimizers of the classical barrier
function.  Before stating this theorem, we state the following
theorem, which will be used in the proof of the theorem on local convergence.
\begin{theorem}\label{thrm-N}
  Consider problem (\ref{eqn-inequality}), where $f:\Re^n\rightarrow
  \Re$ and $c:\Re^n\rightarrow \Re^m$ are continuous.  Let
  $\mathcal{N}$ denote the set of all local constrained minimizers
  with objective function value $f^*$, and assume that $f^*$ has been
  chosen so that $\mathcal{N}$ is nonempty.  Assume further that the
  set $\mathcal{N}^*\subseteq \mathcal{N}$ is a nonempty compact
  isolated subset of $\mathcal{N}$.  Then there exists a compact set
  $\mathcal{S}$ such that $\mathcal{N}^*$ lies in
  $\interior(\mathcal{S})\cap\mathcal{F}$ and for any feasible point
  $\bar{x}$ in $\mathcal{S}$ but not in $\mathcal{N}^*$, $f(\bar{x}) >
  f^*$.  Furthermore, every point $x^*$ in $\mathcal{N}^*$ has the
  property that $f(x^*)=f^*=\min f(x)$ for all $x\in\mathcal{S}\cap\mathcal{F}$.
\end{theorem}
\begin{proof}
See {\cite[Theorem 7]{FiaM90}} and {\cite[Theorem 6]{Wri92}}.
\end{proof}

The following theorem and proof regarding convergence of the minimizers of the
barrier function is found in \cite{ForGW02}.

\begin{theorem}\label{thrm-local-primal}
  Consider problem (\ref{eqn-inequality}), where $f:\Re^n\rightarrow
  \Re$ and $c:\Re^n\rightarrow \Re^m$ are continuous.  Let
  $\mathcal{F}$ denote the feasible region, let $\mathcal{N}$ denote
  the set of minimizers with objective function value $f^*$, and
  assume that $\mathcal{N}$ is nonempty.  Let $\{\mu_k\}$ be a
  strictly decreasing sequence of positive barrier parameters such
  that $\lim_{k\to\infty} \mu_k = 0$.  Assume that
\begin{enumerate}
 \item[\rm(a)] there exists a nonempty compact set $\mathcal{N}^*$ of local
      minimizers that is an isolated subset of $\mathcal{N}$;

 \item[\rm(b)] at least one point in $\mathcal{N}^*$ is in the closure of
      $\strict(\mathcal{F})$.
\end{enumerate}
Then the following results hold:
\begin{enumerate}
\item[\rm (i)] there exists a compact set $\mathcal{S}$ such that
$\mathcal{N}^*\subset\interior(\mathcal{S})\cap\mathcal{F}$ and such that, for any feasible point
$\bar{x}$ in $\mathcal{S}$ but not in $\mathcal{N}^*$, $f(\bar{x}) > f^*$;

\item[\rm (ii)] for all sufficiently small $\mu_k$, there is an
\emph{unconstrained} minimizer $y_k$ of the barrier function $B_{\mu_k}(x)$ in
$\strict(\mathcal{F})\cap\interior(\mathcal{S})$, with
\[
   B_{\mu_k}(y_k) = \hbox{\rm min}\,
   \{\, B_{\mu_k}(x) \st  x\in \strict(\mathcal{F})\cap \mathcal{S}\, \}.
\]
Thus $B_{\mu_k}(y_k)$ is the smallest value of $B_{\mu_k}(x)$ for
any $x\in \strict(\mathcal{F})\cap \mathcal{S}$;

\item[\rm (iii)] any sequence of these unconstrained minimizers
$\{y_k\}$ of $B_{\mu_k}(x)$ has at least one convergent subsequence;

\item[\rm (iv)] the limit point $x_\infty$ of any convergent subsequence
$\{x_k\}$ of the unconstrained minimizers $\{y_k\}$ defined in
{\rm (ii)} lies in $\mathcal{N}^*$;

\item[\rm (v)] for the convergent subsequences $\{x_k\}$ of
part {\rm (iv)},
\[
  \lim_{k\to\infty}f(x_k) =f^* = \lim_{k\to\infty} B_{\mu_k}(x_k).
\]
\end{enumerate}
\end{theorem}
\begin{proof}
See~\cite{ForGW02}.
\end{proof}
%note the above is a local convergence theorem 
%could converge to points that are not local constrained minimizers.

Applied to problem (\ref{eqn-inequality}), the classical barrier
method can be viewed as a path-following method that defines a path to
both an optional $x^*$ and associated Lagrange multipliers $\lambda^*$.

The following theorem (stated and proved in~\cite{ForGW02} and based
on results from \cite{FiaM90, Wri92, SWriO99}) summarizes the
conditions under which a sequence of barrier minimizers converges to
the solution of (\ref{eqn-inequality}). 

\begin{theorem}\label{thrm-barrier}
Consider problem (\ref{eqn-inequality}).  Assume that the set
of strictly feasible points is nonempty.
Let $x^*$ be a local constrained minimizer of (\ref{eqn-inequality}),
$\nabla f(x)\defined g(x)$,
$J(x)\defined \nabla c(x)^T$, and $\mathcal{A}$
denote the set of indices of the active constraints at $x^*$.  Assume that the
following sufficient optimality conditions hold at $x^*$:
\begin{itemize}
\item[(a)] $x^*$ is a KKT point, i.e., there exists a nonempty set
  $\mathcal{M_\lambda}$ of Lagrange multipliers $\lambda$ satisfying 
$$M_\lambda = \{ \lambda \, : \, g(x^*) = J(x^*)^T\lambda, \,\, 
\lambda\ge 0, \,\, \text{and}  \,\,  c_i(x^*) \lambda_i = 0 \,\,\, \forall i \};$$
\item[(b)] there exists $p$ such that $J_{\mathcal{A}}(x^*) p >0$,
  where $J_{\mathcal{A}}(x^*)$ denotes the active constraints at
  $x^*$; and 
\item[(c)] there exists $\omega >0$ such that $p^T H(x^*,\lambda) p \ge
  w\|p\|_2^2$ for all $\lambda\in \mathcal{M_\lambda}$ and all nonzero
  $p$ satisfying $g(x^*)^Tp = 0$ and $J_{\mathcal{A}} (x^*) p\ge 0$,
  where $H(x^*,\lambda)\defined \nabla^2 f(x^*)-\sum_{i=1}^m
  \lambda_i\nabla^2c_i(x^*)$ is the Hessian of the Lagrangian
  evaluated at $x=x^*$.
\end{itemize}
Assume that a logarithmic barrier method is applied in which $\mu_k$
converges monotonically to zero as $k\rightarrow \infty$.  Then
\begin{itemize}
\item[(i)] there is at least one subsequence of unconstrained
  minimizers of the barrier function $B_{\mu_k}(x)$ converging to $x^*$;
\item[(ii)] let $\{x^k\}$ denote such a convergent subsequence.  Then
  the sequence of barrier multipliers $\{\lambda_k\}$, whose 
  $i$th
  component is $\mu_k/c_i(x^k)$ is bounded;
\item[(iii)] $\lim_{k\rightarrow\infty} \lambda(x^k) = \bar{\lambda}\in
  \mathcal{M_\lambda}$.
\end{itemize}
If, in addition, strict complementarity holds at $x^*$, i.e., there is
a vector $\lambda\in \mathcal{M_\lambda}$ such that $\lambda_i>0$ for
all $i\in \mathcal{A}$, then
\begin{itemize}
\item[(iv)] $\bar{\lambda}_{\mathcal{A}} > 0;$
\item[(v)] for sufficiently large $k$, the Hessian matrix $\nabla^2
  B_{\mu_k}(x^k)$ is positive definite;
\item[(vi)] a unique, continuously differentiable vector function
  $x(\mu)$ of unconstrained minimizers of $B_{\mu}(x)$ exists for
  positive $\mu$ in a neighborhood of $\mu=0$; and 
\item[(vii)] $\lim_{\mu\rightarrow 0_+} x(\mu) = x^*$.
\end{itemize}
\end{theorem}
\begin{proof}
See~\cite{ForGW02}.
\end{proof}

Consider the case when $c(x)\defined x$, and the constrained local
minimizer $x^*$ is strictly positive.  In this case,
Theorem~\ref{thrm-barrier} reduces to the following corollary:

\begin{corollary}
Consider problem (\ref{eqn-inequality}) with $c(x)\defined x$.
Suppose $x^*$ be a local constrained minimizer of
(\ref{eqn-inequality}).  Further, assume that $x^*>0$.  Let $\nabla f(x)$, 
$J(x)$ and $\mathcal{A}$ be defined as in Theorem \ref{thrm-barrier}.
Assume that the following sufficient optimality conditions hold at $x^*$:
\begin{itemize}
\item[(a)] $x^*$ is a stationary point of $f(x)$, i.e., $g(x^*)=0$;
\item[(b)] $\nabla^2 f(x^*)$ is positive definite.
\end{itemize}
Assume that a logarithmic barrier method is applied in which $\mu_k$
converges monotonically to zero as $k\rightarrow \infty$.  Then
\begin{itemize}
\item[(i)] there is at least one subsequence of unconstrained
  minimizers of the barrier function $B_{\mu_k}(x)$ converging to $x^*$;
\item[(ii)] let $\{x^k\}$ denote such a convergent subsequence.  Then
  the sequence of barrier multipliers $\{\lambda(x^k)\}$, whose 
  $i$th
  component is $\mu_k/c_i(x^k)$ is bounded;
\item[(iii)] $\lim_{k\rightarrow\infty} \lambda(x^k) = 0$;
\item[(iv)] for sufficiently large $k$, the Hessian matrix $\nabla^2
  B_{\mu_k}(x_k)$ is positive definite;
\item[(v)] a unique, continuously differentiable vector function
  $x(\mu)$ of unconstrained minimizers of $B_\mu(x)$ exists for
  positive $\mu$ in a neighborhood of $\mu=0$; and 
\item[(vi)] $\lim_{\mu\rightarrow 0_+} x(\mu) = x^*$.
\end{itemize}
\end{corollary}
\begin{proof}
  This proof is based on the proof given in~\cite{ForGW02} for
  Theorem~\ref{thrm-barrier}, modified for the case when $c(x)=x$
  and the constrained local minimizer $x^*$ is strictly positive.
  
  Since $x^*>0$, $x^*$ is in $\strict(\mathcal{F})$, and thus, in the
  closure of $\strict(\mathcal{F})$.  Assumptions (a) and (b) imply
  that $x^*$ is an isolated \emph{unconstrained} minimizer of $f(x)$.
  Thus, the conditions of Theorem~\ref{thrm-local-primal} are met,
  implying that there is at least one subsequence of unconstrained
  minimizers of $B_{\mu_k}(x)$ converging to $x^*$.  This proves (i).

Let $\{x^k\}$ denote such a convergent
sequence, i.e., $\lim_{k\to\infty} x^k = x^*$.  Each $x^k$ is an
unconstrained minimizer of $B_{\mu_k}(x)$:
\begin{equation*}
  \nabla B_{\mu_k}(x)=g(x^k)-\sum_{i=1}^m \nabla c_i(x^k) \lambda_i(x^k),\quad\hbox{where}\quad
  \lambda_i(x^k) = \frac{\mu_k}{c_i(x^k)}.
\end{equation*}
Because $c(x^k) > 0$ (from Theorem~\ref{thrm-local-primal}, result (ii)),
$\lambda_i(x^k)$ is strictly positive for any $\mu_k > 0$.  Since
there are no active constraints at $x^*$ and 
$x^k$ converges to $x^*$,
\begin{equation*}           
   \lim_{k\to\infty} c_i(x^k) = c_i(x^*) > 0,\quad\hbox{and hence}\quad
   \lim_{k\to\infty} \lambda_i(x^k) =  0,
\end{equation*}
for all $i=1,\ldots, m$, proving (ii) and (iii).

As in the proof of (v) in Theorem~\ref{thrm-barrier},
to determine the properties of $\nabla^2
B_{\mu_k} (x^k)$ as $k\to\infty$,
we write the Hessian of the barrier function (\ref{eqn-frisch}) as
\begin{equation*}
   \nabla^2 B_{\mu_k} (x^k) = \nabla^2 f(x^k) + 
        -\sum_{i=1}^m {\lambda_i(x^k)}\; \nabla^2 c_i(x^k) + \sum_{i=1}^m \frac{\lambda_i(x^k)}{c_i(x^k)}  \nabla c_i(x^k)
            (\nabla c_i(x^k))^T.
\end{equation*}
Since $x^k\to x^*$ and
$\lambda(x^k)\to 0$ as $k\to\infty$, $\lim_{k\to\infty} \nabla^2
B_{\mu_k} (x^k)  =  \nabla^2 f(x^*)$.  Thus, by assumption (b), for
sufficiently large $k$, the Hessian  $\nabla^2 B_{\mu_k}(x^k)$ is
positive definite, proving (iv).

To verify the existence of a unique, differentiable function $x(\mu)$
for positive $\mu$ in a neighborhood of $x(\mu_k)$, we apply the implicit function
theorem (see, for example, \cite[p.~128]{OrtR00} and \cite[pp.~585--586]{NocW99})
to the $n+1$ variables $(x,\mu)$.  At $(x^k,                                               
\mu_k)$, we know that the following system of nonlinear equations has
a solution:
\[
  \Phi(x,\mu) = g(x)
\]
The Jacobian of $\Phi$ with respect to $x$
is the barrier Hessian $\nabla^2 B_{\mu}(x)$, which was just shown to be
positive definite at $x=x^k$ and $\mu=\mu_k$.  The implicit function
theorem then implies that there is a locally unique,
differentiable function
$x(\mu)$ passing through $x(\mu_k)\defined x^k$ such that $\Phi(x(\mu),\mu) = 0$
for all positive $\mu$ in a neighborhood of $\mu_k$.  

Using continuation arguments, it is straightforward to show that the
function $x(\mu)$ exists for all $0<\mu\le \mu_k$ for all sufficiently
large $k$, giving result (vi).

Result (vi) is immediate from the local uniqueness of $x(\mu)$ and
result (i), that $x_k$ is a local unconstrained minimizer of the
barrier function.
\end{proof}

\section{The Primal Barrier Energy Method}
\label{sec:primal}

Inequality constraints on $u$ introduced in Section \ref{sec:einstein}
can be enforced using a barrier method.   Define
$$J_{\mu}(u)=J(u)-\mu\int_\Omega \ln (u) \,\,\, dx,$$
where $J(u)$ is defined as in (\ref{eqn-energy}).  The Gateaux derivative
is given by 
\begin{eqnarray}
J'_\mu(u)(v) & = & J'(u)(v)-\frac{d}{dt}\left[ \mu\int_{\Omega} \ln
  (u+tv)\right]\bigg|_{t=0} \nonumber \\ 
& = & \int_{\Omega}\left[ \left(\bar{a}\nabla u\right)\cdot\nabla v
+k(u)v \right] \,\,\, dx \nonumber\\
& & +\int_{\partial_N\Omega}(c u - g_N) v\,\,\, ds - \mu\int_{\Omega}
u^{-1}v \,\,\, dx.\label{eqn-energyderiv}
\end{eqnarray}
Thus, the condition for stationarity of $J_{\mu}(u)$ is given by solving the
following problem:
\begin{equation}\label{eqn-stat-mu}
\text{Find}~ u\in H_{0}^1{(\Omega)} ~\text{s.t.}~
J'_{\mu}(u)(v)=0, 
\ \  \forall v\in H_{0}^1{(\Omega)}.
\end{equation}

The second Gateaux derivative of $J_{\mu}(u)$ is given by
\begin{equation}
J''_{\mu}(u)(w,v)= \int_{\Omega} \left[ \left( \bar{a}\nabla w\right)\cdot \nabla v + k'(u)wv\right]  \,\,\, dx
+ \int_{\partial_N\Omega}  cwv \,\,\, ds + \mu\int_{\Omega} u^{-2}wv\,\,\, dx.
\label{eqn-secondmu}
\end{equation}
Thus, the Newton update for solving (\ref{eqn-stat-mu}) is given by:
\begin{equation}\label{eqn-newtonConst}
\text{Find}~ w\in  H_{0,D}^1(\Omega) ~\text{s.t.}~
J''_{\mu}(u)(w,v) = - J'_{\mu}(u)(v), 
\ \ \forall v\in H_{0,D}^1(\Omega).
\end{equation}

\subsection{Discretization}  We use a standard Galerkin finite element method to approximate the solution (\ref{eqn-newtonConst})
in an $N$-dimensional subspace $X_h\subset X = H_{0,D}^1(\Omega) \cap [u_-,u_+]$.
Thus, we seek a solution $u_h\in X_h$ such that
\begin{equation}\label{eqn-stat-mu-sub}
J'_{\mu}(u_h)(v_h)=0, 
\ \  \forall v\in X_h.
\end{equation}
The Newton update is given by:
\begin{equation}\label{euler-discretized}
\text{Find}~ w_h\in X_h\subset X ~\text{s.t.}~
J''_{\mu}(u_h)(w_h,v_h) = - J'_{\mu}(u_h)(v_h), 
\ \ \forall v_h\in X_h.
\end{equation}
Let $\{\phi_i\}_1^N$ be a basis $X_h$.  Then, without loss of
generality, let
$$u_h=\sum_{i}^N \alpha_i\phi_i, \quad w_h= \sum_{i}^N \beta_i\phi_i,$$
for some $\{\alpha_i\}$ and $\{\beta_i\}$.  It is sufficient to take the test functions $v_h \in
X_h$ to be the basis functions $\{\phi_i\}_{i=1}^N$.
Equation (\ref{euler-discretized}) is equivalent to solving the following
matrix-vector equation:
\begin{equation}\label{eqn-mat-vec-primal}
\left[A(u_h)+\mu M(u_h)\right]W=-\left[G(u_h)-\mu H(u_h)\right],
\end{equation}
where 
\begin{eqnarray}
A_{ij}(u_h) & =  &J''_\mu(u_h)\left(\phi_j,\phi_i\right), \\
M_{ij}(u_h) & = & \int_\Omega \left(u_h\right)^{-2}\phi_j\phi_i\,\,\, dx\\
W_i & = & \beta_i, \\
G_i(u_h) & = & J'_\mu(u_h)(\phi_i),\\
H_i(u_h) & = & \int_{\Omega} (u_h)^{-1}\phi_i \,\,\, dx
\end{eqnarray}
The barrier term contributes an extra term to the system matrix, namely
$$ M=\mu\int_{\Omega} u_h^{-2}\phi_j\phi_i\,\,\, dx. $$
Thus, the barrier term adds a positive definite matrix to the original
system matrix, and so, may be viewed as a regularization.  In \FETK,
this integral is approximated using a high-accuracy quadrature rule,
using a finite sum with fixed positive weights.  

The Newton update $W$ defines a descent direction for
\begin{equation} \label{eqn-primalMerit}
\phi(u_h)=\frac{1}{2}\| G(u_h) \|_2^2,
\end{equation}
and thus, $\phi(u_h)$ may be used as a merit function to enforce
sufficient descent. 

At optimality, the solution $u_h$ must lie in the strict interior of
the feasible region, and thus, the Lagrange multipliers must be
exactly zero.  Because of this, there is no restriction that $\mu$ must
be kept away from zero.  Thus, $\mu$ may
be steadily decreased, and in fact, may be set to zero--solving the
original stationary problem.

Algorithm~\ref{alg-barrier-primal} summarizes the primal barrier energy method:

\begin{figure}[ht]
\begin{center}
\begin{minipage}{\textwidth}
%%%%%%%%%%%%%%%%%%%%%%%%%%%%%%%%%%%%%%%%%%%%%%%%%%%%%%%%%%%%%%
\begin{Pseudocode}{Algorithm~\ref{alg-barrier-primal}: Primal Barrier Energy Method}\label{alg-barrier-primal}
\tab Choose $u_0>0$, $\mu>0$, $\eta\in \left(0,\frac{1}{2}\right)$, and $\gamma\in (0,1)$;
\tab Set $k=0$; 
\tab \WHILE{} $\NOT \converged$ \DO{}
\tab Compute $u(\mu)$ to approximately solve (\ref{eqn-stat-mu-sub}):
\tabb \WHILE{} $\NOT \converged$ \DO{}
\tabbb Solve (\ref{eqn-mat-vec-primal}) to obtain $w^i_h$;
\tabbb Use the ``99\% rule'' to compute $\alpha_{\max}$;
\tabbb Compute $\alpha\in(0,\alpha_{max}]$ such that $\phi(u^i_h+\alpha w^i_h) \leq \phi(u^i_h)+
\eta\alpha\nabla \phi(u_h^i)^Tw^i_h$;
\tabbb  $u^i_h\leftarrow u^i_h+\alpha w^i_h$; 
\tabb \ENDDO{}
\tabb $u_{k+1}\leftarrow u(\mu)$;
\tabb $\mu\leftarrow \gamma \mu$;
\tabb $k\leftarrow k+1$; 
\tabb $i\leftarrow 0$;
\tab \ENDDO{} 
\end{Pseudocode}
%%%%%%%%%%%%%%%%%%%%%%%%%%%%%%%%%%%%%%%%%%%%%%%%%%%%%%%%%%%%%%%%%%%%%                   
\end{minipage}
\end{center}
\end{figure}

In practice, each barrier function $B_\mu (x)$ does not have to be minimized to
high precision.  Typically, each barrier function is considered
sufficiently minimized when the norm of its gradient is either less
than an fixed absolute tolerance or satisfies a relative tolerance
based on $B_\mu (x_0)$, where $x_0$ denotes the initial
guess for the minimization (see, for example, \cite{DemES82}).

%%%%%%%%%%%%%%%%%%%%%%%%%%%%%%%%%%%%%%%%%%%%%%%%%%%%%%%%%%%%%%%%%%%%%%%%%%%%%%
\section{Numerical Results}
\label{sec:num}

The standard Newton method, a standard Newton method with safeguarding, and
the primal barrier energy method was implemented using \FETK{} 
(the Finite Element ToolKit; see~\cite{Hols2001a} 
and \href{http://www.FETK.org}{http://www.FETK.org}).
These methods were used to solve the Einstein
constraint equations on three single-hole domains, centered at the origin,
with given boundary conditions on both the inner and outer boundary.  
Each tetrahedral mesh was generated by the GAMer component of FETK,
which is a high-fidelity surface and volume meshing tool based on
standard simplex triangulation, subdivision, and smoothing algorithms
(cf.~\cite{HoYu07a,HoYu07d}).
Details of the three meshes are given in Table~\ref{table-meshes}.

\begin{table}[ht]
\caption{\label{table-meshes} Meshes}
%refined to have at least 1000 verts and simplicies: mesh#1=mc405, mesh#2=408, mesh #3=mc410 (altered to robin)
\begin{center} \footnotesize
\begin{tabular}{lccc}\hline
\hstrt                        & Mesh \#1 &     Mesh \#2         &  Mesh \#3 \\\hline
%\strt Center             &  $(0,0,0)$ &    $(0,0,0)$       & $(0,0,0)$ \\
\strt Inner radius    &   50          &    10           &  1 \\
\strt   Outer radius   &  100         &     100         &  100  \\
\strt   Vertices          &   2089           &     1436         &   2820  \\
\hstrt  Simplicies     &  9726           &     7589       &  15321  \\ \hline
\end{tabular}
\end{center}
\end{table}

At the heart of each nonlinear solver is a linear solver (e.g., sparse
direct solver or the conjugate-gradient (\CG) method).  For more
ill-conditioned systems, the Newton equations may not be solved exactly;
however, it is necessary that any step obtained from the linear solver of
choice must be a descent direction.  A simple backtracking line search is
used to obtain a step that meets the sufficient decrease criteria in
Algorithms~\ref{alg-barrier} and~\ref{alg-barrier-primal}.  Convergence is
obtained when the norm of the nonlinear residual defining the PDE 
is less than a chosen tolerance of $\epsilon\defined$1.0e-07, i.e.,  
\begin{equation}\label{eqn-termination}
\|G(u_k)\|_2 \leq \epsilon . 
\end{equation}

For each example, the energy barrier method initialized $\mu$ and then
decreased $\mu$ whenever the iterate $u^k_h$ satisfied 
\begin{equation} \label{eqn-subproblem}
\|f(u^k_h)\|_2\leq \max\{\epsilon_\mu \|f(u^0_h)\|_2, \epsilon_\mu\}
\end{equation}
where $f(u^k_h)=G(u^k_h)-\mu H(u^k_h)$, $u^0_h$ denotes the initial iterate
after decreasing $\mu$, and $\epsilon_\mu=\max\{\min\{0.1,\mu\}, \epsilon\}$.
This choice of $\epsilon_\mu$ allows each subproblem to be solved to greater
accuracy as $\mu$ is decreased.

Reasonable choices for parameters for the Lichnerov equation with boundary
conditions include those given in Examples 1--2, in
Section~\ref{chap:energy}.  The first two examples are with these choices
of parameters.  
\newline 
\newline 
\textbf{Example 1}.  Set $\bar{a}=1.0$,
$R=1.0$, $\tau=0.1$, $\sigma=0.2$, $\rho=0.1$, $c=1.0$, and $g_N=-1.0$.  We
pick $\Omega$ to have Robin boundary conditions on both the inner and outer
boundary.  The presence of the negative exponents acts as a natural barrier
function, preventing infeasible iterates.  For this reason, to obtain a
positive solution it is sufficient to add a safeguarding procedure such as
the so-called ``99\% rule'' (see Section~\ref{sec:methods}) to the standard
Newton method.

Table~\ref{table-11} gives the results of using the standard Newton method,
Newton's method with safeguarding using the ``99\% rule'', and the primal
energy method to solve the Lichnerovich equation.  For each solver, the
initial guess was a vector of all ones.  We list the number of iterations
(``itns''), residual, and the signs of the entries in the vector of
coefficients for $u_h$ (i.e., $+$ denotes all the entries are strictly
positive, $-$ denotes all the entries are strictly negative, and $+/-$
denotes both positive and negative entries).  Note that there is one linear
solve per iteration, and thus, the number of iterations is also the number
of linear solves required by each method.  All four solvers converged on
all three meshes to strictly positive solutions.

In Table~\ref{table-11}, the energy barrier method is reported with two
different initial values of $\mu$.  First, the energy barrier method was
run with $\mu=0.0$, making it numerically equivalent to Newton's
method with safeguarding.  (For nonconvex problems, we do not expect the
energy barrier method to converge to a strictly positive solution with this
choice of $\mu$).  For illustrative purposes, the results with $\mu=1.0$
are also reported; for this test, when each subproblem was sufficiently
solved (i.e., (\ref{eqn-subproblem}) was satisfied), we reduced $\mu$ by a
factor of $1/10$.  However, for this convex problem, allowing faster
reductions of $\mu$ leads to fewer overall iterations.  In fact, reducing
$\mu$ by a factor of $1/100$ led to 14 iterations on mesh \#1, 16
iterations on mesh \#2, and 16 iterations on mesh \#3.  For more difficult
problems, reducing $\mu$ too quickly will impede convergence.  Even though
faster reductions in $\mu$ would lead to results more similar to the Newton
methods, for the results in Table~\ref{table-11}, we chose to display
results with a reasonable, commonly accepted reduction of $1/10$.  Also, it
is worth noting that the different residual values for the algorithms are
inconsequential in that convergence only requires (\ref{eqn-termination})
to be satisfied.  (In these cases, Newton's method was fortunate in that
quadratic convergence led to a much smaller residual for each final
iterate.)  

In Table~\ref{table-11} we see that a positive solution was obtained by all
methods using an initial guess of all ones.  As previously noted, a barrier
method is not required to obtain a strictly positive solution since the
energy functional has a natural barrier in the form of negative
coefficients of $u$.  Thus, we expect that a safeguarded Newton method will
be sufficient to recover a positive solution.  In this example, Newton's
method with and without safeguarding produced the same iterates.  It is of
interest to note that a strictly negative solution can be recovered on all
three meshes by using the standard Newton method together with the initial
guess of a vector of all negative ones.

\begin{table}[ht]
  \caption{\label{table-11}  \footnotesize Example \#1.}
\begin{center}
\begin{footnotesize}
\begin{tabular}{|l|ccc|ccc|ccc|}\hline
\multicolumn{1}{|c|}{ }&\multicolumn{3}{c|}{\hstrt Mesh \#1}&\multicolumn{3}{c|}{Mesh \#2}&\multicolumn{3}{c|}{Mesh \#3 } \\\hline
\hstrt method &  itns & resid & sign &  itns & resid & sign & 
 itns & resid & sign   \\ \hline
\strt Newton (standard)        &  6& 2.71e-12&  + & 6& 5.15e-12& + & 6& 4.73e-12&  +\\ 
\strt Newton (safeguarded)     &  6& 2.71e-12&  + & 6& 5.15e-12& + & 6& 4.73e-12&  +\\ 
\strt Barrier energy ($\mu_0=0.0$) &  6& 2.71e-12&  + & 6& 5.15e-12& + & 6& 4.73e-12&  +\\ 
\strt Barrier energy ($\mu_0=1.0$)  & 22& 6.15e-08&  + &24& 1.19e-08& + & 24& 1.17e-08& + \\\hline
\end{tabular}
\end{footnotesize}
\end{center}
\end{table}

%%%%%%%%%%%%%%%%%%%%%%%%%%%%%%%%%%%%%%%%%%%%%%%%%%%%%%%%%%%%%%%%%%%%%%%%%%%%%
%%%%%%%%%%%%%%%%%%%%%%%%%%%%%%%%%%%%%%%%%%%%%%%%%%%%%%%%%%%%%%%%%%%%%%%%%%%%%
%%%%%%%%%%%%%%%%%%%%%%%%%%%%%%%%%%%%%%%%%%%%%%%%%%%%%%%%%%%%%%%%%%%%%%%%%%%%%
\noindent
\textbf{Example 2}.  Let $a=2$, $R=-1000$, $\tau=\sqrt{72}$,
$\sigma=\sqrt{48}$, and $\rho=1/\pi$.  For this example, the inner and outer
boundaries of all three meshes have a Robin boundary condition with $c=2$ and $g_N=10$.

On the first mesh, the standard Newton method failed to converge in 50
iterations, denoted by the asterisk in Table~\ref{table-22}; in fact,
at the 100th iteration, the current approximation to $u_h$ contained
both positive and negative entries--the standard Newton method was
unable to maintain a strictly positive solution.  When the initial
guess was set to a vector whose entries were all 10, Newton's method
converged to a positive solution; when the initial guess was set to a
vector whose entries were all -10, the standard Newton method returned
a strictly negative solution; and when the initial guess was set to a
vector whose enries were all -1, the standard Newton method returned a
strictly positive solution.  (Newton's method without safeguarding was
unpredictable--one could recover a positive solution even when
starting with an initial negative guess.) Meanwhile, with an initial
guess of all ones, the safeguarded Newton method converged to a
strictly positive solution.  The barrier energy method converged with
$\mu_0=1$ and subsequent reductions in $\mu$ of $1/10$, as in Example 1.

On the second mesh, all methods converged to the same strictly positive
solution when the initial guess was all ones.  (The barrier energy method
was run with $\mu_0=1$ and subsequent reductions in $\mu$ of $1/10$, as in Example 1.)  A
strictly negative solution can be recovered by using the standard Newton
method together with the initial guess of a vector whose entries are all
-10.

On the third mesh, Newton's method without safeguarding did not
converge within the first 100 iterations and had both positive and
negative components.  A negative solution was obtained by the standard
Newton method by starting with a vector whose entries were all -10.
Meanwhile, Newton's method with safeguarding made insignificant
progress after the seventh iteration, and thus, failed to converge in
100 iterations. The safeguarded method failed because the initial
safeguarded Newton steps took some of the coefficients of $u_h$ very
close to the boundary and the following Newton directions continued to
point in the direction of negative numbers for these components.  In
this event, the ``99\% rule'' allows subsequent steps of only
negligible size along the Newton direction at each iteration.  As a
result, the safeguarded Newton method is unable to make any real
progress each subsequent iteration, and thus, fails to converge.

In this example, we see that even though there is a natural barrier in
the Lichnerovich equation, there are additional benefits that a
barrier function approach can offer: By setting the parameter $\mu$ to
be large enough, we can alter the pure Newton direction, preventing
the initial Newton iterates from getting too close to the boundary.
For example, setting $\mu_0$ to be 10, 20, 30, or 40, the energy
method's iterates took large initial steps to the boundary, preventing
convergence; however, with $\mu_0=50$, the algorithm converged to a
positive solution.  (The values in Table~\ref{table-22} are with
$\mu_0=50$ and subsequent reductions in $\mu$ of $1/10$.)

The energy barrier method was the only method to find a strictly positive
solution to this problem on all three meshes.

\begin{table}[ht]
  \caption{\label{table-22}  \footnotesize Example \#2.}
\begin{center}
\begin{footnotesize}
\begin{tabular}{|l|ccc|ccc|ccc|}\hline
\multicolumn{1}{|c|}{ }&\multicolumn{3}{c|}{\hstrt Mesh \#1}&\multicolumn{3}{c|}{Mesh \#2}&\multicolumn{3}{c|}{Mesh \#3 } \\\hline
\hstrt method &  itns & resid & sign &  itns & resid & sign & 
 itns & resid & sign   \\ \hline
\strt Newton (standard)        &  *& 4.69e+19& +/-&11& 9.80e-08& +/- & *& 4.92e+20&  +/- \\ 
\strt Newton (safeguarded)     &  9& 2.89e-09&  + & 7& 1.37e-08& + & *& 2.01e+07&  +\\ 
\strt Barrier energy           & 16& 4.02e-08&  + &16& 7.67e-08& + &17& 9.47e-09& + \\\hline
\end{tabular}
\end{footnotesize}
\end{center}
\end{table}

%%%
\noindent
\textbf{Example 3}.  Consider following Yamabe problem:
$$
-8\Delta u + \rho (r) u^5=0,
$$
where $\rho(r)=1/r^3$, and $r$ is the Euclidean distance. (This choice of
$\rho(r)$ was motivated by equation (41) in~\cite{bmp09}).  For this example, we impose
the Dirichlet condition $u=1$ on both the inner and outer boundaries of all
three meshes.  For this example, the initial guess was taken to be a vector
of ones; given the Dirichlet condition, this is a reasonable starting
point.

Table~\ref{table-33} reports the results of each solver on this problem.
Neither backtracking nor a barrier method was required to solve this problem.

\begin{table}[ht]
  \caption{\label{table-33}  \footnotesize Example \#3.}
\begin{center}
\begin{footnotesize}
\begin{tabular}{|l|ccc|ccc|ccc|}\hline
\multicolumn{1}{|c|}{ }&\multicolumn{3}{c|}{\hstrt Mesh \#1}&\multicolumn{3}{c|}{Mesh \#2}&\multicolumn{3}{c|}{Mesh \#3 } \\\hline
\hstrt method &  itns & resid & sign &  itns & resid & sign & 
 itns & resid & sign   \\ \hline
\strt Newton (standard)        & 1 & 1.91e-08&  + & 2& 1.16e-12& + & 3& 1.18e-12 & +\\ 
\strt Newton (safeguarded)     & 1 & 1.91e-08&  + & 2 &1.16e-12& + & 3 & 1.18e-12& +\\ 
\strt Barrier energy ($\mu$=1.0) & 18& 1.13e-12&  + &22&1.25e-12& + & 23& 1.34e-12& +\\ \hline
\end{tabular}
\end{footnotesize}
\end{center}
\end{table}

%%%
\noindent
\textbf{Example 4}.  For this example, we modify the problem in Example 3 to include the extra term in (\ref{eqn-yamabe}):
$$
-8\Delta u - \frac{1}{8} u= - \rho (r) u^5,
$$
where $\rho(r)=1/r^3$, and $r$ is the Euclidean distance.  Also, assume
Dirichlet boundary conditions of $u=1$ on both the inner and outer
boundaries.  For this example, the initial guess was taken to be a vector
of ones.

Table~\ref{table-44} contains the results on all three meshes.  The
standard Newton method converged quickly to a solution with positive and
negative components on each mesh.  Newton's method with safeguarding did
not converge on any mesh.  On the all three meshes, the method took initial
large steps to the boundary and made negligible progress after ten
iterations. The primal barrier energy method with $\mu_0=1$ converged on
the first mesh, but this initial value of $\mu$ on the second mesh was too
small.  With $\mu_0=10$, the primal barrier energy method converged to a strictly positive solution on all three meshes.

\begin{table}[ht]
  \caption{\label{table-44}  \footnotesize Example \#4.}
\begin{center}
\begin{footnotesize}
\begin{tabular}{|l|ccc|ccc|ccc|}\hline
\multicolumn{1}{|c|}{ }&\multicolumn{3}{c|}{\hstrt Mesh \#1}&\multicolumn{3}{c|}{Mesh \#2}&\multicolumn{3}{c|}{Mesh \#3 } \\\hline
\hstrt method &  itns & resid & sign &  itns & resid & sign & 
 itns & resid & sign   \\ \hline
\strt Newton (standard)        & 11 & 2.71e-08&  +/- & 16& 9.65e-09& +/- & 18& 8.97e-11 & +/- \\ 
\strt Newton (safeguarded)     & * & 9.16e+03&  + & * & 1.64e+04& + & * & 1.47e+04& +\\ 
\strt Barrier energy ($\mu$=10.0) & 17& 1.94e-11&  + &18&2.86e-11& + & 18& 2.69e-11& +\\ \hline
\end{tabular}
\end{footnotesize}
\end{center}
\end{table}

%%%%%%%%%%%%%%%%%%%%
\section{Conclusion}
\label{sec:conc}

In this article we considered both the design and the analysis of a
certain class of nonconvex optimization-based numerical methods for
approximating positive solutions to nonlinear geometric elliptic
partial differential equations containing critical exponents.  As
noted, these types of problems arise regularly in geometric analysis
and mathematical physics; our primary interest here was Yamabe problem
and the Einstein constraint equations.  The difficulty one faces with
these problems are the simultaneous presence of several challenging
features, including spatial dimension $n \geqs 3$, varying and
potentially non-smooth coefficients, critical (or even super-critical)
nonlinearity, non-monotone nonlinearity (arising from a non-convex
energy), and spatial domains that are typically Riemannian manifolds
rather than simply open sets in $R^n$.  For these types of problems,
there may be multiple solutions, although only positive solutions
typically have mathematical and physical meaning.  This creates
additional complexities in both the theory and numerical treatment of
such problems, as this feature introduces both non-uniqueness as well
as the need to incorporate an inequality constraint into the
formulation.  As a practical approach for treating these difficulties,
we considered numerical methods based on Galerkin-type discretization,
covering any standard bases construction (finite element, spectral, or
wavelet), and the combination of a barrier method for nonconvex
optimization and global inexact Newton-type methods for dealing with
nonconvexity and the presence of inequality constraints.  After giving
an overview of barrier methods in non-convex optimization, we then
developed and analyzed a primal barrier energy method.  We then
presented a sequence of numerical experiments using this type of
barrier method, based on a particular Galerkin method, namely the
piecewise linear finite element method, leverage the FETK modeling
package.  In the experiments, the negative pole in the Hamilitonian
constraint provided a ``natural'' barrier, aiding the convergence of
Newton methods.  In this setting, a barrier method will often be
unnecessary; however, in some cases the numerical experiments showed
that allowing for a flexible barrier parameter can be helpful (see
Example 2, Mesh \#3).  The experiments also confirmed that on some
classes of Yamabe problems, a solution could not be found without the
use of a barrier method (see Example \#4), suggesting that barrier
methods are more useful on critical exponent problems without
singularities that arise from lower-order negative exponent terms.

Although we considered here only scalar elliptic equations with variational
structure, that is, they that arise as the Euler condition for stationarity 
of an underlying (usually nonconvex) energy, more generally these types of 
critical exponent problems may arise as part of a more complex elliptic system.
A prime example is the coupled Hamiltonian and momentum constraints in
the Einstein equations~\cite{HNT07a,HNT07b}.
While the Hamiltonian constraint (as well as the momentum constraint)
alone has variational structure, when combined as a system there is in
fact no variational structure to exploit.
Nevertheless, the ideas in this paper can be applied by using
alternative formulations of the Einstein constraints, and will
be pursued in a second article.

%%%%%%%%%%%%%%%%%%%%%%%%%%%%%%%%%%%%%%%%%%%%%%%%%%%%%%%%%%%%%%%%%%%%%%%%%%%%%%
\section*{Acknowledgments}
   \label{sec:ack}

MH was supported in part by NSF Awards~0715146 and 0915220,
and by DOD/DTRA Award HDTRA-09-1-0036.
JE was supported in part by a Ralph E. Powe Junior Faculty Enhancement Award.

\bibliographystyle{abbrv}
\bibliography{mjh,jbe}
%\clearpage
%\input{app}

%\vspace*{0.5cm}

\end{document}